\documentclass[letterpaper,11pt]{article}

\AtBeginDocument{%
  }

\usepackage[letterpaper, margin=1in]{geometry}
\usepackage{amsmath,amsthm,amsfonts}
\usepackage[utf8]{inputenc}
\usepackage[margintag,small,off]{our-comments}
\usepackage{xspace}
\usepackage{tcolorbox}
\tcbuselibrary{skins,breakable}
\usepackage{enumitem}
\usepackage{thmtools}
\usepackage{xfrac}
\usepackage{hyperref}
\usepackage{cleveref}
\newcommand{\defn}[1]{\textbf{\emph{#1}}}

\renewcommand{\epsilon}{\varepsilon}

\newcommand{\poly}{\operatorname{poly}}

\newcommand{\prob}[1]{\Pr\left[#1\right]}
\DeclareMathOperator{\E}{\mathbb{E}}

\newcommand{\expect}[1]{\E\left[#1\right]}

\newcommand{\ceiling}[1]{\left\lceil #1 \right\rceil}

\newcommand\numberthis{\addtocounter{equation}{1}\tag{\theequation}}

\newcommand{\calA}{\mathcal{A}}
\newcommand{\calB}{\mathcal{B}}
\newcommand{\calR}{\mathcal{R}}
\newcommand{\calX}{\mathcal{X}}
\newcommand{\calY}{\mathcal{Y}}
\newcommand{\LRU}{\text{LRU}\xspace}

\newcommand{\OPT}{\text{OPT}\xspace}
\newcommand{\universe}{\mathcal{U}}
\newcommand{\local}[1]{\langle{#1}\rangle}
\newcommand{\fflocal}[1]{\langle{#1}\rangle^{\text{FF}}}
\newcommand{\iflocal}[1]{\langle{#1}\rangle^{\text{IF}}}
\newcommand{\ways}{\alpha}
\newcommand{\evict}{\textsf{Out}}

\newtheorem{theorem}{Theorem}
\newtheorem{lemma}{Lemma}
\newtheorem{corollary}{Corollary}
\newtheorem{proposition}{Proposition}
\newtheorem{claim}{Claim}

\creflabelformat{equation}{#2\textup{(#1)}#3}
\crefname{equation}{eq.}{eq.}
\Crefname{equation}{Eq.}{Eq.}

\newenvironment{tbox}{\begin{tcolorbox}[
		enlarge top by=5pt,
		enlarge bottom by=5pt,
		 boxsep=0pt,
                  left=4pt,
                  right=4pt,
                  top=10pt,
                  arc=0pt,
                  boxrule=1pt,toprule=1pt,
                  colback=white
                  ]
}
{\end{tcolorbox}}

\addCommenter{guido}{Guido}{Guido}{red}
\addCommenter{rathish}{Rathish}{Rathish}{blue}
\addCommenter{mfc}{MFC}{mfc}{purple}

\begin{document}

\title{An Associativity Threshold Phenomenon in Set-Associative Caches}

\date{}

\author{
Michael A. Bender\thanks{Stony Brook University. \href{mailto:bender@cs.stonybrook.edu}{\texttt{bender@cs.stonybrook.edu}}}
\and
Rathish Das \thanks{University of Liverpool. \href{mailto:rathish.das@liverpool.ac.uk}{\texttt{rathish.das@liverpool.ac.uk}}}
\and
Mart\'{\i}n Farach-Colton \thanks{Rutgers University. \href{mailto:martin@farach-colton.com}{\texttt{martin@farach-colton.com}}}
\and
Guido Tagliavini \thanks{Rutgers University. \href{mailto:guido.tag@rutgers.edu}{\texttt{guido.tag@rutgers.edu}}}
}

\maketitle

\begin{abstract}
In an $\ways$-way set-associative cache, the cache is partitioned into disjoint sets of size $\ways$, and each item can only be cached in one set, typically selected via a hash function. Set-associative caches are widely used and have many benefits, e.g., in terms of latency or concurrency, over fully associative caches, but they often incur more cache misses. As the set size $\ways$ decreases, the benefits increase, but the paging costs worsen.

In this paper we characterize the performance of an $\ways$-way set-associative LRU cache of total size $k$, as a function of $\ways = \ways(k)$. We prove the following, assuming that sets are selected using a fully random hash function:
\begin{itemize}
    \item For $\ways = \omega(\log k)$, the paging cost of an $\ways$-way set-associative LRU cache is within \emph{additive} $O(1)$ of that a fully-associative LRU cache of size $(1-o(1))k$, with probability $1 - 1/\poly(k)$, for all request sequences of length $\poly(k)$.
    \item For $\ways = o(\log k)$, and for all $c = O(1)$ and $r = O(1)$, the paging cost of an $\ways$-way set-associative LRU cache is \emph{not} within a \emph{factor} $c$ of that a fully-associative LRU cache of size $k/r$, for some request sequence of length $O(k^{1.01})$.
    \item For $\ways = \omega(\log k)$, if the hash function can be occasionally changed, the paging cost of an $\ways$-way set-associative LRU cache is within a \emph{factor} $1 + o(1)$ of that a fully-associative LRU cache of size $(1-o(1))k$, with probability $1 - 1/\poly(k)$, for request sequences of arbitrary (e.g., super-polynomial) length.
\end{itemize}

Some of our results generalize to other paging algorithms besides LRU, such as least-frequently used (LFU).
\end{abstract}
\section{Introduction}

In an \defn{$\ways$-way set-associative cache}, the cache is partitioned into disjoint sets, each of size $\ways$, and there is an \defn{indexing function}~\cite{Topham1999RandomizedCache} that selects the unique set into which each item can be placed. The set size $\ways$ is also known as the \defn{associativity} of the cache.

Set-associative caches are extensively used in hardware~\cite{Aingaran2015M7, Konstadinidis2016M7, Kurd2014Haswell, Bowhill2015Xeon, Tam2018SkyLakeSP, Rotem2022AlderLake}, because they are faster, use simpler circuitry, and consume less power~\cite{Hill1988DirectMapping, Hill1989Associativity, Przybylski1988Tradeoffs, Kessler1989SetAssociativity} than \defn{fully associative caches}, which can store any item in any slot in the cache. Set-associative caches are also used as a building block in the design of highly-concurrent software caches~\cite{Adas2022LimitedAssociativity, RocksDB2022LRUCache}, because the sets in a set-associative cache are independent and can operate in parallel. The smaller the set size $\ways$ is, the bigger the benefits are.

However, the smaller the set size $\ways$, the worse the \defn{paging cost}, that is, the number of cache misses. This is because an item can only be placed in one of the cache slots in its allocated set, so when more than $\ways$ frequently accessed items land in the same set, they contend for the slots in the set, causing increased paging.

Thus, when a set-associative cache is being designed, $\ways$ must be chosen small enough to provide the desired latency, power consumption, chip area or concurrency; and $\ways$ must be large enough to keep the paging cost low.  How can a cache designer optimize for $\ways$ given this tradeoff?  Measuring the benefits of a small $\ways$ depends on details that are beyond the scope of this paper.  Here, we focus on the other side of the tradeoff: \textit{what is the increase in paging cost of an $\ways$-way set-associative cache, relative to a fully associative cache, as a function of $\alpha$?}

Answering this question is a well-known open problem in the theory~\cite{Smith1976SetAssociativity, Smith1978SetAssociativity, Sen2013Models, Xiang2013HOTL} and practice~\cite{Hill1989Associativity, Przybylski1988Tradeoffs, Agarwal1988CachePerformance, Bell1974CacheOrganizations} of set-associative caches. The only prior theoretical work that partially answered the question did so under the assumption that item requests come from a distribution~\cite{Smith1976SetAssociativity, Smith1978SetAssociativity, Sen2013Models, Xiang2013HOTL}; their bounds do not hold against an adversary. Proving a worst-case bound in the adversarial setting has remained open.

Existing theory literature has also analyzed alternative types of caches with limited associativity, namely companion caches~\cite{Brehob2003RestrictedCaching, Mendel2004CompanionCaching, Buchbinder2014RestrictedCaching} and two-level caches~\cite{Peserico2003ArbitraryAssociativity}. However, because these analyses compare the paging cost against the optimal algorithm for the limited-associative cache in consideration, they do not provide any information about the paging cost relative to a fully associative cache, so they do not help the cache designer pick the right associativity to trade off benefits against paging costs.

\subsection{This paper}

In this paper we solve the open problem of analyzing the paging cost of a set-associative cache as compared to a fully associative cache, with a focus on the case of set-associative least-recently used (LRU) caches. In these caches, when an item is fetched but the designated set is full, the least-recently accessed item from that set is evicted. We assume that the indexing function $h$ is a hash function that selects sets uniformly and independently at random.

Our main results establish a remarkable threshold phenomenon in set-associative caches.  Specifically, if a cache of size $k$ has associativity $\ways = o(\log k)$, then the paging costs are large, even compared to that of a fully-associative cache whose size is only a small fraction of $k$; and if $\ways = \omega(\log k)$, the paging costs are similar to that of a fully-associative cache of roughly the same size $k$. Thus, on the one hand, cache designers should make sure the associativity is at least logarithmic, and, on the other hand, they should be aware that increasing the associativity beyond logarithmic offers diminishing returns.

\paragraph*{Competitive analysis for long request sequences}

Let $k$ be the size of the cache. We show that if $\ways = \omega(\log k)$, then $\ways$-way set-associative LRU is $1$-competitive with fully associative LRU with probability $1 - 1/\poly(k)$, using $(1 + \Theta(\sqrt{\log (k) / \ways}))$-resource augmentation, on any request sequence of length at most $\poly(k)$.\footnote{In this paper, $\poly(k)$ means \emph{any} polynomial in $k$, as long as we choose the hidden constants appropriately.} (See \Cref{sec:preliminaries} for the formal definitions of competitiveness and resource augmentation.)

In fact, this analysis follows from a more general theorem that establishes a tradeoff among all the parameters involved, namely the set size, the resource augmentation factor, the success probability and the length of the longest supported request sequence. Moreover, the theorem works for a class of paging algorithms that we call \defn{stable}. Roughly speaking, when a stable algorithm is run on a small cache and is given only part of the request sequence, its paging decisions approximate that of the same algorithm running on a large cache and given the whole request sequence. Besides LRU, this class contains least-frequently used (LFU), and a generalization of LRU, known as LRU-$K$~\cite{ONeil1993LRUK}. This paper introduces and studies the class of stable algorithms, which may be of independent interest.

\paragraph*{The limits of set-associative caches}
Notice that the aforementioned competitive analysis holds for all request sequences up to certain length (which is a function of the parameters involved), but not \emph{arbitrarily long} request sequences. Is our analysis to blame for this or is it a limitation of set-associative caches? More generally, can any of the parameters of the analysis be improved? That is, can we prove $1$-competitiveness (or even $O(1)$-competitiveness) on sub-logarithmic associativity, or on super-polynomial input sequences, or using less than $1 + \Theta(\sqrt{\log(k) / \ways})$ resource augmentation?

We show that the answer to all of these questions is essentially negative, via a more general tradeoff. We prove that given a set size $\ways$, a resource augmentation factor $r$, and a competitive ratio $c$, and provided $r$ is not extremely large, there exists an adversarial request sequence that is not too long (as a function of $\ways$, $r$ and $c$), on which $\ways$-way set-associative LRU is \emph{not} $c$-competitive with fully associative LRU using $r$-resource augmentation. Using this tradeoff we establish that our analysis is tight: 
\begin{itemize}
    \item \textit{Arbitrary length:} Set-associative LRU is \emph{not} competitive on arbitrarily long request sequences.
    \item \textit{Associativity:} $\ways$-way set-associative LRU is \emph{not} $c$-competitive with fully associative LRU using $r$-resource augmentation on request sequences of length as small as $O(k^{1.01})$, for any $\ways = o(\log k)$, $r = O(1)$ and $c = O(1)$.
    \item \textit{Resource augmentation:} $\ways$-way set-associative LRU is \emph{not} $c$-competitive with fully associative LRU using $r$-resource augmentation on request sequences of length as small as $O(k^{2.01})$, for any $\ways = \Omega(\log k)$, $r = 1 + o(\sqrt{\log(k)/\ways})$ and $c = O(1)$.
    \item \textit{Super-polynomial length:} $\ways$-way set-associative LRU is \emph{not} $c$-competitive with fully-associative LRU using $r$-resource augmentation on request sequences of length $\ell$, for any $\ways = \Omega(\log k)$, $r = 1 + \Theta(\sqrt{\log(k)/\ways})$, $c = \poly(k)$ and super-polynomial $\ell$.
\end{itemize}

\noindent
In fact, these negative results hold not only for LRU, but for all conservative paging algorithms, which includes FIFO, LFU and clock.

\paragraph*{Competitive analysis for arbitrarily long request sequences} Set-associative caches fail on long enough input sequences because it does not take too long for an adversary (even one that is oblivious to the random choices of the indexing function) to produce a malicious collection of items that are unevenly distributed into the sets. Such a collection will cause some set to become oversubscribed, thus forcing unnecessary cache misses due to the associativity constraints; these are called \defn{conflict misses}~\cite{Hill1989Associativity}. Repeatedly querying the items from one of these malicious collections drives the miss count up in the set-associative cache, but, importantly, not in a fully associative cache (as long as the collection is smaller than the size of the cache), so the performance gap can be arbitrarily large.

To circumvent this issue, we use \defn{rehashing}: We equip the set-associative cache with the ability to flush all the items and draw a new hash function. We show that rehashing every $\poly(k)$ \emph{cache misses} makes set-associative LRU almost $1$-competitive with fully associative LRU, on request sequences of arbitrary length. Specifically, for $\ways = \omega(\log k)$, $\ways$-way set-associative LRU with rehashing is $(1+1/\poly(k))$-competitive with fully associative LRU with probability $1 - 1/\poly(k)$, using $(1 + \Theta(\sqrt{\log(k)/\ways}))$-resource augmentation, on request sequences on arbitrary length. Furthermore, we show that the rehashing can be done gradually, attaining the same performance guarantees as a stop-the-world-type rehashing.

\subsection{Organization}

The paper is organized as follows. In \Cref{sec:related} we describe the related work. In \Cref{sec:preliminaries} we review basic facts about paging algorithms, and present the terminology and probability tools that we will use in the rest of the paper. In \Cref{sec:local} we present the competitive analysis of set-associative caches for long (but not arbitrarily long) request sequences. In \Cref{sec:lower-bounds} we establish the performance limits of set-associative caches. In \Cref{sec:long} we augment set-associative LRU with the ability to rehash, and analyze it on arbitrarily long request sequences. Finally, in \Cref{sec:stable-algorithms} we investigate stable algorithms in depth.

\section{Related work}
\label{sec:related}

The foremost algorithmic technique to analyze the paging cost is competitive analysis~\cite{Borodin1998Online}. There is a long line of work on competitive analysis of paging algorithms, which mostly focuses on fully associative algorithms~\cite{Sleator1985Paging, Fiat1991Marking, Boyar2007Paging, Dorrigiv2008OnlineAnalysis, Dorrigiv2009Paging, Young1994LooseCompetitiveness, Young2002FileCaching, AgrawalBeDa21, AgrawalBeDa20, DasAB20, DeLayoZA22, agrawal2022online}.

In the reduced-associativity setting (sometimes called \textit{restricted caching}), existing work focuses mostly on two types types of caches, namely \textit{two-level caches}~\cite{Peserico2003ArbitraryAssociativity} and \textit{companion caches}~\cite{Brehob2003RestrictedCaching, Mendel2004CompanionCaching, Buchbinder2014RestrictedCaching}, which are generalizations of set-associative caches. Two-level caches are, essentially, set-associative caches where items can be mapped into multiple sets, and copied between any of them for a fraction of the cost of a cache miss. This means that, occasionally, it may be convenient to make room in a set by moving one of its items to another set, instead of evicting it. Companion caches are a combination of a set-associative cache (the \textit{main} cache) with a, typically smaller, fully associative cache (the \textit{companion}).\footnote{In the hardware architecture literature, companion caches are known as \textit{victim caches}~\cite{Jouppi1990VictimCache}.} In every case, the aim is to design online algorithms that are competitive against an optimal offline algorithm for the specific cache organization at hand. The resulting algorithms are not necessarily practical---they are complex, and may not meet the extremely tight latency requirements of hardware caches.

Recent work has considered paging algorithms that are subject to the placement rules of set-associative caches, but not their eviction rules---pages can be evicted from any set, at any time. The idea is to have a set-associative cache mirror the paging decisions of a fully associative paging algorithm, which is simulated on the side. Every time the fully associative algorithm evicts a page, the set-associative cache evicts the same page, even if the the set the victim is occupying is underfull. Crucially, the set-associative cache is resource-augmented with respect to the simulation, so it is unlikely that a set ever becomes full. Using this technique, Bender et al.~\cite{Bender2021AddressTranslation} showed that every online fully associative algorithm can be made to work with the associativity restrictions of set-associative caches. The downside is that simulating a fully associative algorithm is computationally expensive. In a very different context, Frigo et al.~\cite[Lemma 16]{Frigo1999CacheOblivious} proved a similar simulation result specifically for LRU, albeit using larger resource augmentation.

As the associativity decreases, the performance gap between set- and full-associativity widens. In the extreme case of $1$-way set-associative caches, also known as \textit{direct-mapped caches}, Agrawal et al.~\cite{Agrawal2009WorstPaging} showed that their paging cost has the undesired property of being a linear factor away from that of the \emph{worst} offline fully associative algorithm, in expectation.

Before competitive analysis became the de facto standard for analyzing online algorithms, probabilistic models were the prominent tool. For the analysis of storage systems, the most popular one is the \textit{independent reference model} (IRM)~\cite{Coffman1973OS, King1971DemandPaging, Rao1978CachePerformance}, in which every item in the system is accessed independently with some fixed probability. Using the IRM, Smith~\cite{Smith1976SetAssociativity, Smith1978SetAssociativity} compared the probability of a cache miss in a fully associative LRU and in a set-associative LRU, assuming an access distribution called \textit{linear paging model}. He concluded that the two probabilities are asymptotically equal, as the size of the sets grows. Of course, because the IRM is a distributional model of the input, IRM analyses do not hold for adversarial workloads, and their conclusions strongly depend on the choice of the access distribution.

Other lines of research include analyzing the paging cost of set-associative caches for specific workloads~\cite{Sanders1999SetAssociativity, Fix2003SetAssociativityTrees, Harper1999NestedLoops}, and emulating external-memory algorithms~\cite{Aggarwal1988ExternalMemory} that were originally designed for fully associative caches in set-associative caches~\cite{Sen2002TheoryCacheEfficient}.
\section{Preliminaries}
\label{sec:preliminaries}

\paragraph{Paging algorithms}

A paging algorithm $\calA$ is parametrized by the size of the cache it operates on. Given a positive integer $k$, we denote $\calA_k$ the algorithm $\calA$ using a cache of size $k$. Let $\universe$ be the universe of items. Let $\sigma \in \universe^*$ be a sequence of item requests. The \defn{cost} of $\calA_k$ on $\sigma$, denoted $C(\calA_k, \sigma)$, is defined as the number of cache misses that $\calA_k$ incurs throughout $\sigma$.

We denote $\sigma_i$ the $i$-th element in $\sigma$. For any $x \in \universe$, we denote $\sigma x$ the sequence that results from appending $x$ to $\sigma$. For any $X \subseteq \universe$, we write $\sigma[X]$ the subsequence of $\sigma$ that results from removing every request that is not for an item in $X$. The notation $x \in \sigma$ means that item $x$ is present in the request sequence $\sigma$.

We denote $\calA_k(\sigma)$ the set of items cached by $\calA_k$ after serving the sequence $\sigma$. When the input sequence $\sigma$ is clear from the context, we denote $\calA_k(i)$ the set of items cached by $\calA_k$ right after $\sigma_i$ is served. Denote $\evict(\calA_k, \sigma, x) := \calA_k(\sigma) \setminus \calA_k(\sigma x)$ the set of items evicted by $\calA_k$ in response to an access to $x$ right after the sequence $\sigma$ has been served.

\paragraph{Competitive analysis}

We compare paging algorithms using competitive analysis~\cite{Sleator1985Paging}. A paging algorithm $\calA$ is \defn{$c$-competitive} with paging algorithm $\calB$ \defn{with probability $p$} using \defn{$r$-resource augmentation} on request sequences of \defn{length $\ell$} if, for every request sequence $\sigma$ with $|\sigma| \leq \ell$,
\[
\prob{C(\calA_k, \sigma) \leq c \cdot C(\calB_{k'}, \sigma) + O(1)} \geq p,
\]
where $k = rk'$. Here, the probability is taken over the randomness of the algorithms. In this paper, the randomness typically comes from the indexing hash function. We say that $\calA$ is \defn{not $c$-competitive} with $\calB$ using $r$-resource augmentation on request sequences of length $\ell$ if $\calA$ is \emph{not} $c$-competitive with $\calB$ with probability $1/2$ using $r$-resource augmentation on request sequences of length $\ell$. We will always assume that the input $\sigma$ is oblivious to the random bits used by the algorithms, that is, $\sigma$ is produced by an \defn{oblivious adversary}. All paging algorithms in this paper are \defn{online}, that is, each request arrives only after the previous one has been served by the algorithm.

\paragraph{Classes of paging algorithms}

Among the prominent deterministic paging algorithms are least-recently used (LRU), least-frequently used (LFU), first-in first-out (FIFO), and clock. A well-known generalization of LRU, called LRU-$K$~\cite{ONeil1999LRUKOptimality, ONeil1993LRUK, Boyar2010LRUK}, always evicts the item whose $K$-th most recent access is the least recent. In particular, regular LRU is LRU-$1$.\footnote{LRU-$2$ often outperforms LRU~\cite{ONeil1993LRUK}, because it is less sensitive to isolated accesses, as it only deems an item important when it has been recently accessed at least two times.} A paging algorithm is \defn{lazy}~\cite{Manasse1990Server} if it fetches an item from slow memory only when it causes a cache miss, evicts at most one item upon a cache miss, and only evicts an item when the cache is full. LRU-$K$, FIFO, LFU and clock all fall into this category. For comparison, flush-when-full~\cite{Karlin1986FWF} is a non-lazy algorithm that flushes the whole cache when an item is fetched but every slot is taken. A paging algorithm is \defn{conservative}~\cite{Young1992Thesis} if it incurs at most $k$ cache misses on any consecutive subsequence of the input that contains at most $k$ distinct items. LRU, FIFO, LFU and clock are all conservative; flush-when-full is \emph{not} conservative.

\paragraph{Concentration bounds}

We will use the following standard variant of the Chernoff bound (see, e.g., ~\cite{Devdatt2009Concentration}).

\begin{theorem}
\label{thm:chernoff}
Let $X_1, \dots, X_n$ be negatively associated $0$/$1$ random variables. Let $X := \sum_i X_i$. For every $\epsilon \in (0, 1)$,
\[
    \prob{X \geq (1 + \epsilon)\expect{X}}\leq \exp(-\epsilon^2 \expect{X} / 3),
\]
and
\[
\prob{X \leq (1 - \epsilon)\expect{X}} \leq \exp(-\epsilon^2 \expect{X} / 2).
\]
\end{theorem}

We will also use the well-known fact that Chernoff bounds are tight.

\begin{theorem}[{\cite{Mousavi2012Chernoff}}]
\label{thm:reverse-chernoff}
Let $X_1, \dots, X_n$ be independent $0$/$1$ random variables with $\prob{X_i = 1} = p \leq 1/2$ for every $i$. Let $X := \sum_i X_i$. For every $\epsilon \in [0, 1/p - 2]$,
\[
\prob{X \geq (1+\epsilon)\expect{X}} \geq \frac{1}{4} \exp(-2 \epsilon^2 \expect{X}).
\]
\end{theorem}

\section{Set-Associative Caches Are Competitive on Long Sequences}
\label{sec:local}

We model set-associative and fully associative caches as paging algorithms. For any paging algorithm $\calA$ and any $\ways = \ways(k)$ that divides $k$, we define \defn{$\ways$-way set-associative $\calA$} as the following paging algorithm.

\begin{tbox}
\textbf{Algorithm:} $\ways$-way set-associative $\calA$ with cache size $k$
\begin{itemize}
    \item Pick a fully random hash function $h : \universe \to [k/\ways]$.
    \item Divide the cache into $k / \ways$ sets of size $\ways$ slots each, and run a separate instance of $\calA_{\ways}$ on each set.
    \item When an item $x$ is requested, execute the request on set $h(x)$.
\end{itemize}
\end{tbox}

We denote by $\local{\calA}$ the $\ways$-way set-associative $\calA$ algorithm. We omit $\ways$ in the notation, as only a single $\ways$ will be in consideration at all times. Because the case $\ways = k$ represents a fully associative cache, we always assume $\ways < k$.  
To avoid confusion between sets in a cache and the mathematical meaning of a set, we will refer to the sets in a set-associative cache as \defn{buckets}.

The goal of this section is to show that the number of cache misses that $\local{\calA}$ incurs is at most approximately the same as that of $\calA$, for any lazy algorithm $\calA$ that belongs to a class that we call \defn{stable algorithms}, which we will define shortly (in \Cref{subsec:stable}). In particular, this class contains the LFU and LRU-$K$ algorithms. At the end of this section we will have proven the following theorem.

\begin{theorem}
\label{thm:monotone-local}
Let $\calA$ be a lazy stable paging algorithm. Let $\delta = \delta(k)$ be such that $\delta \geq \sqrt{12 \ln(k / \ways)/\ways}$ and $\delta \leq 1/2$ for all $k$ large enough.
Then, $\ways$-way set-associative $\calA$ is $1$-competitive  with $\calA$ w.p. $1 - \exp(-\delta^2 \ways/24)$ using $(1 - \delta)^{-1}$-resource augmentation on request sequences of length $O(\exp(\delta^2 \ways / 24))$.
\end{theorem}

\noindent
Roughly speaking, \Cref{thm:monotone-local} says that the paging cost of $\local{\calA}_k$ is approximately the same as that of $\calA_{k'}$, where $k' = (1 - \delta)k$, as long as $\delta$ is not too small, and as long as the request sequence is not too long. The tunable parameter $\delta$ is the capacity gap between the set-associative cache and its fully associative counterpart, and it controls the tradeoff between the resource augmentation factor, the success probability and the length of the request sequence. On the one hand, to minimize the resource augmentation factor $(1 - \delta)^{-1}$, we want $\delta$ to be as small as possible. On the other hand, to maximize the success probability and the sequence length, which are both increasing functions of $\delta$, we want $\delta$ to be as large as possible. Perhaps the most interesting region in this tradeoff curve is in the lower end of $\delta$, which asymptotically minimizes the resource augmentation while still providing polynomial success probability and support for polynomial sequences, as the following result shows.

\begin{proposition}
\label{prp:super-logarithmic}
Let $\calA$ be a lazy stable paging algorithm. Suppose $\ways = \omega(\log k)$. Then, $\ways$-way set-associative $\calA$ is $1$-competitive with $\calA$ w.h.p.\footnote{\textit{With high probability} (w.h.p.) means with probability $1 - 1/\poly(k)$.} using $(1 + \Theta(\sqrt{\log(k) / \ways}))$-resource augmentation on request sequences of length $\poly(k)$.
\end{proposition}
\begin{proof}
For an arbitrary constant $c > 0$, let $\delta := \sqrt{24 c \ln (k) / \ways}$. Because $\ways = \omega(\log k)$, $\delta$ is in the appropriate range for all $k$ large enough. This choice yields success probability $1 - k^{-c}$ and sequence length $O(k^{c})$. The resource augmentation factor satisfies
\[
1 + \delta \leq (1 - \delta)^{-1} \leq 1 + 2\delta,
\]
where the second inequality is because $\delta \leq 1/2$. Thus, $(1 - \delta)^{-1} = 1 + \Theta(\delta) = 1 + \Theta(\sqrt{\log(k)/\ways})$.
\end{proof}

\subsection{Proof Outline}

Before diving into the details of \Cref{thm:monotone-local}, we give an outline of its proof for the particular case of $\calA = \LRU$. Call an item \textit{stale} if it is older (with respect to the LRU order) than every item in the fully associative cache. Intuitively, evicting items that are \emph{not} stale is bad, because an access to them would cause a hit in the fully associative cache but a miss in the set-associative cache. We call an eviction of a non-stale item a \emph{bad} eviction. The proof is divided into two steps: First, we prove that if few evictions are bad, then the cost of the set-associative cache is at most approximately the same as that of the fully associative cache. Second, we prove that the probability that an eviction is bad is low; this implies that the total number of bad evictions is small.

We upper bound the probability that an eviction is bad using the following chain of observations: If an item $x$ is evicted and $x$ is not stale, then none of the pages in $x$'s bucket are stale either. This implies that $x$ as well as all other items in $x$'s bucket are present in the fully associative cache. Therefore, the set of items in the fully associative cache have the undesirable property that, when hashed into the set-associative cache, they cause some bucket to become overfull. Crucially, this set of items is deterministic, so the property is provably unlikely by a balls-and-bins-type analysis of the bucket loads.

\subsection{Stable Paging Algorithms}
\label{subsec:stable}

Notice that in the proof outline we have used the following property of LRU: if a bucket in the set-associative cache evicted an item $x$ that is present in the fully associative cache, then all items in $x$'s bucket are also present in the fully associative cache. Stable algorithms formalize this property.

We say that a paging algorithm $\calA$ is \defn{stable} if, for every $\tau \in \universe^*$, $X \subseteq \universe$, $z \in X$, and positive integers $a > b$,
\begin{equation}
\label{eq:stable}
\text{$\evict(\calA_b, \tau[X], z) \cap \calA_{a}(\tau z) \neq \emptyset$ implies $\calA_{b}(\tau[X] z) \subseteq \calA_{a}(\tau z)$.}
\end{equation}
\noindent
In words, if an instance of $\calA$ is given a request sequence $\tau z$, a smaller instance is given a subsequence $\tau[X] z$, and if in response to the access to the final item $z$ the smaller cache evicts an item that is also present in the larger cache, then the smaller cache is fully contained in the larger cache.

We can also think of stability as a seamless mechanism that corrects the differences between any two caches managed by instances of $\calA$.
Roughly speaking, the contrapositive of \Cref{eq:stable} says that $\calA_b$ always evicts items that are not present in $\calA_a$ first.
Thus, intuitively, stable algorithms are effective at managing the buckets of a set-associative cache because they try to approximate the state of a larger fully associative cache.

We defer a thorough discussion of stable algorithms to \Cref{sec:stable-algorithms}, in which we investigate the scope of this class of algorithms, and provide insight into the kind of features that distinguish paging algorithms that are stable from those that are not. The following fact is a corollary of \Cref{sec:stable-algorithms}.
\begin{lemma}
\label{lem:monotone-algorithms}
LRU-$K$ and LFU are stable. FIFO and clock are \emph{not} stable.
\end{lemma}

\subsection{Auxiliary Results}

Let $\calX$ and $\calY$ be any two paging algorithms, with arbitrary cache sizes.
We say that an item $x$ is a \defn{bad eviction} of $\calX$ with respect to $\calY$ at time $i$ (we omit $\calX$, $\calY$ or $i$ when they are clear from the context) if $x$ is evicted by $\calX$ due to access $\sigma_i$ and $x \in \calY(i)$, that is, $x \in \evict(\calX, \sigma_1 \dots \sigma_{i - 1}, \sigma_i) \cap \calY(i)$.

\begin{lemma}
\label{lem:unstale-evictions}
Let $\calX$ and $\calY$ be any two paging algorithms with arbitrary cache sizes, such that they only fetch a page when it is requested. Then,
\[
C(\calX, \sigma) \leq C(\calY, \sigma) + B,
\]
where $B$ is the number of bad evictions of $\calX$ with respect to $\calY$ throughout $\sigma$.
\end{lemma}
\begin{proof}
Observe that $C(\calX, \sigma) - C(\calY, \sigma) = M - H$, where $H$ is the number of cache hits in $\calX$ that are misses in $\calY$, and $M$ is the number of misses in $\calX$ that are hits in $\calY$. We call the latter \defn{bad misses}. Thus, it suffices to show that $M \leq B$.

Suppose $\sigma_i$ is a bad miss. Notice that $\sigma_i$ must have already been accessed before time $i$, or else it wouldn't be present in $\calY$ at time $i$, as items are only fetched when they are requested. Therefore, $\sigma_i$ must have been evicted by $\calX$ in the past. Let $j < i$ be the most recent time when $\sigma_i$ was evicted by $\calX$. It is enough to show that $\sigma_i$ was a bad eviction (of $\calX$ with respect to $\calY$ at time $j$); this will establish an injection from bad misses into bad evictions, as there cannot be two misses of the same item without being evicted in-between. Thus, to conclude the proof, we need to show that $\sigma_i \in \calY(j)$.

During the time interval $(j, i)$, item $\sigma_i$ was not accessed, or else $j$ wouldn't be the most recent eviction time of $\sigma_i$ in $\calX$. Since $\sigma_i$ is a hit in $\calY$ at time $i$, item $\sigma_i$ was not accessed during $(j,i)$, and items are only fetched upon a request, we must have $\sigma_i \in \calY(j)$. Hence, $\sigma_i$ was a bad eviction of $\calX$ with respect to $\calY$ at time $j$.
\end{proof}

In particular, \Cref{lem:unstale-evictions} holds for lazy algorithms. (Later on, in \Cref{sec:long}, we will use this result in full generality, on paging algorithms that perform multiple evictions during a time step.)

Next, we prove a standard balls-and-bins bound for our particular setting.

\begin{lemma}
\label{lem:balls-and-bins}
Let $\delta = \delta(k)$ be such that $\delta \geq \sqrt{12 \ln(k/\ways) / \ways}$ and $\delta \leq 1/2$ for all $k$ large enough. Suppose $k' = (1-\delta)k$ balls are randomly and independently thrown into $k/\ways$ bins. Then, for all $k$ large enough, the maximum load among all bins is greater than $\ways$ with probability at most $\exp(-\delta^2 \ways/12)$.
\end{lemma}
\begin{proof}
Let $m := k' = (1-\delta)k$ be the number of balls, and let $n := k/\ways$ be the number of bins. Let $h := m/n = (1-\delta)k / (k/\ways) = (1 - \delta)\ways$ be the average load of a bin. Let $L_i$ be the load of bin $i$. This is a sum of $m$ independent binary random variables, each indicating whether a ball lands in bin $i$. By the Chernoff bound from \Cref{thm:chernoff},
\begin{align*}
\prob{L_i \geq (1 + \delta) \expect{L_i}} &\leq \exp\left(-\frac{\delta^2 \expect{L_i}}{3}\right)\\
&= \exp\left(-\frac{\delta^2 (1-\delta) \ways}{3}\right) \tag{as $\expect{L_i} = h = (1 - \delta)\ways$}\\
&\leq \exp\left(-\frac{\delta^2 \ways}{6}\right). \tag{as $1 - \delta \geq 1/2$}
\end{align*}
\noindent
Notice that
\[
(1 + \delta) \expect{L_i} = (1 + \delta)(1 - \delta) \ways = (1 - \delta^2) \ways < \ways.
\]
Hence, the load of bin $i$ is larger than $\ways$ with probability at most $\exp(\delta^2 \ways / 6)$. By a union bound over all $n$ bins, the maximum load across all bins is larger than $\ways$ with probability
\begin{align*}
n\exp(-\delta^2 \ways/6) &= \frac{k}{\ways} \exp(-\delta^2 \ways/6) \tag{by definition of $n$}\\
&= \exp\left(-\frac{\delta^2 \ways}{6} + \ln\left(\frac{k}{\ways}\right)\right)\\
&\leq \exp\left(-\frac{\delta^2 \ways}{12}\right). \tag{as $\delta^2 \ways/6 \geq 2\ln(k / \ways)$, by hypothesis}
\end{align*}
This concludes the proof.
\end{proof}

\subsection{Proof of \Cref{thm:monotone-local}}

For every time step $i$, let $B_i$ be the indicator random variable that is $1$ when $\sigma_i$ causes a bad eviction in $\local{\calA}_k$ with respect to $\calA_{k'}$.

\begin{claim}
\label{clm:unstale-eviction-probability}
We have $\prob{B_i = 1} \leq \exp(-\delta^2 \ways/12)$.
\end{claim}
\begin{proof}
Suppose an item $x$ is evicted by $\local{\calA}_k$ in response to $\sigma_i$. Then, $\sigma_i$ must hash to the same bucket as $x$. Suppose $x$ is a bad eviction. Let $\tau := \sigma_1 \dots \sigma_{i - 1}$. Then, $x \in \evict(\local{\calA}_k, \tau, \sigma_i) \cap \calA_{k'}(i)$. Let $X := \{y \in \universe \mid h(y) = h(x)\}$ be the set of items that hash to $x$'s bucket. Then, $x \in \evict(\calA_{\ways}, \tau[X], \sigma_i) \cap \calA_{k'}(i)$, because $x$ is evicted from the set-associative cache if and only if it is evicted from its bucket. Because $\calA$ is stable, this implies that $\calA_{\ways}(\tau[X]\sigma_i) \subseteq \calA_{k'}(\tau \sigma_i)$, that is, all items in $x$'s bucket are also in $\calA_{k'}(i)$. Thus, if $S_i := \calA_{k'}(i)$, we conclude that
\begin{align*}
\prob{B_i = 1} \leq \Pr[\text{There are $\ways + 1$ items in $S_i$ that hash to the same bucket}].
\end{align*}

To upper bound this probability we model the event as a balls-and-bins game, in which each ball represents an item, and each bin represents a bucket. Crucially, the set $S_i$ is oblivious to the randomness of $\local{\calA}_k$, so the ball placements are independent and uniformly random. There are at most $|S_i| \leq k'$ balls and $k/\ways$ bins, so by \Cref{lem:balls-and-bins} this probability is at most $\exp(-\delta^2 \ways/12)$.
\end{proof}

Let $B := \sum_{i = 1}^{|\sigma|} B_i$ be the number of bad evictions; there is at most $1$ bad eviction at each step, since $\calA$ is lazy. By \Cref{clm:unstale-eviction-probability} and linearity of expectation, $\expect{B} \leq |\sigma|  \exp(-\delta^2 \ways/12)$. By Markov's inequality, $\prob{B \geq |\sigma| \exp(-\delta^2 \ways/24)} \leq \exp(-\delta^2 \ways/24)$.

Because $\calA$ is lazy, we may apply \Cref{lem:unstale-evictions} to conclude that
\[
C(\local{\calA}_k, \sigma) \leq C(\calA_{k'}, \sigma) + |\sigma| \exp(-\delta^2 \ways/24),
\]
with probability $1 - \exp(-\delta^2 \ways/24)$. When $|\sigma| = O(\exp(\delta^2 \ways/24))$, this implies that $\local{\calA}_k$ is $1$-competitive against $\calA_{k'}$. This concludes the proof.

\section{Set-Associative Caches Are Not Competitive on Sequences of Arbitrary Length}
\label{sec:lower-bounds}

The major limitation of \Cref{thm:monotone-local} is that it does not work for request sequences of arbitrary length. That is, for any valid choice of $\ways$ and $\delta$, competitiveness of the set-associative cache is only guaranteed up to a certain number of requests. A natural question is whether we can extend this result to arbitrarily long input sequences, perhaps by allowing for a larger resource augmentation factor, or for a larger competitive ratio.

The main result of this section is a negative answer to this question for all conservative algorithms, which includes LRU.

\begin{theorem}
\label{thm:tradeoff}
Let $\calA$ be a conservative paging algorithm. Suppose that $|\universe| \geq 16 k \exp(8(1 - \delta)^{-1} \delta^2 \ways)$ and $(1 - \delta)^{-1} \leq k/(2\ways)$ for all large enough $k$. Then, $\ways$-way set-associative $\calA$ is \emph{not} $c$-competitive with $\calA$ using $(1 - \delta)^{-1}$-resource augmentation on request sequences of length
\[
O(c \ways k \exp(16 (1 - \delta)^{-1} \delta^2 \ways)).
\]
\end{theorem}

Roughly speaking, \Cref{thm:tradeoff} states that as long as the resource augmentation $r = (1-\delta)^{-1}$ is not too large and the universe $\universe$ is large enough, there is a hard limit on the length of the request sequences that a set-associative cache can handle before its performance drops.

\begin{corollary}
\label{cor:tradeoff-1}
Let $\calA$ be a conservative paging algorithm. Suppose $\universe$, $\ways$ and $r = (1 - \delta)^{-1}$ satisfy the hypotheses of \Cref{thm:tradeoff}. Then, for every $c \geq 1$, $\ways$-way set-associative $\calA$ is \emph{not} $c$-competitive with $\calA$ using $r$-resource augmentation on \emph{arbitrarily long} request sequences.
\end{corollary}

\noindent
This corollary applies, for instance, to all $\ways \leq k/8$ and $r \leq 4$, provided that the universe is large enough. We can further use \Cref{thm:tradeoff} to show that competitiveness is impossible in some common regimes of $\ways$, $r$ and $c$. For comparison, recall that $1$-competitiveness on polynomial request sequences is possible for $\ways = \omega(\log k)$ and $r = 1 + \Theta(\sqrt{\log(k)/\ways})$, by \Cref{prp:super-logarithmic}.

\begin{proposition}
\label{prp:tradeoff}
Let $\calA$ be a conservative paging algorithm. Suppose (for simplicity) that $\universe$ is infinite. In the following cases, $\ways$-way set-associative $\calA$ is \emph{not} $c$-competitive with $\calA$ using $r$-resource augmentation on request sequences of length $O(\ways k^{1 + o(1)}) = O(k^{2.01})$:
\begin{enumerate}
    \item for $\ways = \Omega(\log k)$ such that $\ways \leq k/3$, $r = 1 + o(\sqrt{\log(k)/\ways})$ and $c = O(1)$;
    \item for $\ways = o(\log k)$, $r = O(1)$ and $c = O(1)$;
    \item for $\ways = 1$, $r = o(\log k)$ and $c = O(1)$.
\end{enumerate}
\end{proposition}
\begin{proof}
The proof of all three cases is similar, so we will only prove the first one. Let $\delta := 1 - 1/r$. Notice that $(1 - \delta)^{-1} = r$. Because $\ways = \Omega(\log k)$, we have $r = 1 + o(1)$ and, therefore, $r = 1 + o(1) \leq 3/2 \leq k/(2\ways)$, so we are in the conditions of \Cref{thm:tradeoff}. Thus, $\ways$-way set-associative $\calA$ is not $c$-competitive with $\calA$ using $r$-resource augmentation on request sequences of length
\begin{align*}
O(c\ways k \exp(16 r \delta^2 \ways)) &= O(\ways k \exp(16 r \delta^2 \ways)) \tag{as $c = O(1)$}\\
&= O(\ways k \exp(o(\log k))) \tag{as $\delta = o(\sqrt{\log(k)/\ways})$ and $r = 1 + o(1)$}\\
&= O(\ways k^{1 + o(1)}).
\end{align*}
\end{proof}

Interestingly, part 3 in \Cref{prp:tradeoff} implies that direct-mapped caches (i.e., set-associative caches with $\ways = 1$) are \emph{not} constant-competitive with a fully associative cache using sub-logarithmic resource augmentation on request sequences as small as $O(k^{1 + o(1)})$. This justifies the conventional wisdom that the performance of direct-mapped caches is subpar.

Parts 1 and 2 in \Cref{prp:tradeoff} imply that \Cref{prp:super-logarithmic} is optimal, in the sense that the set size and resource augmentation factor cannot be shrunk, even if we permit competitive ratio $O(1)$ and sub-polynomial success probability (but at least $1/2$). The following proposition shows that, in fact, the polynomial length of the request sequences in \Cref{prp:super-logarithmic} cannot be improved either, even if the competitive ratio is polynomial.

\begin{proposition}
Suppose that $\universe$ is infinite. For $\ways = \Omega(\log k)$ such that $\ways \leq k/3$, $c = \poly(k)$ and $r = 1 + \Theta(\sqrt{\log(k)/\ways})$, $\ways$-way set-associative $\calA$ is \emph{not} $c$-competitive with $\calA$ using $r$-resource augmentation on request sequences of super-polynomial length.
\end{proposition}

\noindent
The proof is analogous to that of \Cref{prp:tradeoff}, so we omit it.

\subsection{An Auxiliary Result}

Consider a balls-and-bins game, in which $m$ balls are thrown randomly and independently into $n$ bins. Let $h := m/n$. A bin is said to be \defn{$a$-saturated} if, at the end of the game, the bin contains at least $h + a$ balls. Let $f(n, m, \epsilon) := n \exp(-2 \epsilon^2 h)$.

\begin{lemma}
\label{thm:balls-and-bins-lower-bound}
Suppose $m$ balls are thrown into $n \geq 2$ bins. Let $h := m/n$. For every $\epsilon \in [0, n - 2]$, the number of $\epsilon h$-saturated bins is more than $f(n, m, \epsilon)/8$ with probability at least $1 - \exp(-f(n, m, \epsilon)/32)$.
\end{lemma}
\begin{proof}
Let $X_i$ be an indicator random variable that is $1$ exactly when bin $i$ is $\epsilon h$-saturated. Let $L_i$ be the load of bin $i$. Then, $\prob{X_i = 1} = \prob{L_i \geq h + \epsilon h}$. The random variable $L_i$ is a sum of $m$ i.i.d. $0$/$1$ random variables, each one indicating whether a ball lands in bin $i$. In particular, $\expect{L_i} = h$. Because each ball has probability $1/n \leq 1/2$ of landing in bin $i$, and $\epsilon \in [0, n-2]$, we have, by the reverse Chernoff bound from \Cref{thm:reverse-chernoff}, that
\[
\prob{L_i \geq h + \epsilon h} \geq \frac{1}{4} \exp(-2\epsilon^2 h).
\]
Let $X := \sum_i X_i$ be the number of $\epsilon h$-saturated bins. By linearity of expectation, $\expect{X} = \sum_i \expect{X_i} \geq f(n, m, \epsilon) / 4$.

Using standard results about negatively associated random variables~\cite{Wajc2017NA}, one can show that the random variables $X_1, \dots, X_n$ are negatively associated. Thus, by the Chernoff bound,
\begin{align*}
\prob{X \leq f(n, m, \epsilon)/8} &\leq \prob{X \leq \expect{X}/2}\\
&\leq \exp(-\expect{X}/8) \leq \exp(-f(n, m, \epsilon)/32).
\end{align*}
\end{proof}

\subsection{Proof of \Cref{thm:tradeoff}}

Consider the request sequence produced by the following adversary.

\begin{tbox}
\textbf{Adversary:} Request sequence of \Cref{thm:tradeoff}
\begin{enumerate}
    \item Let $s := 16 \exp(8(1 - \delta)^{-1} \delta^2 \ways)$ and $t := c \ways s^2$.
    \item Pick $S_1, \dots, S_{s} \subseteq \universe$ \emph{disjoint} sets of $k' = (1-\delta)k$ items.
    \item For $i = 1, \dots, s$:
    \begin{enumerate}
        \item Repeat $t$ times:
        \begin{enumerate}
            \item Sequentially access every item in $S_i$.
        \end{enumerate}
    \end{enumerate}
\end{enumerate}
\end{tbox}

\noindent
Notice that the adversary can carry out step 2 because the universe is large enough, by assumption. Let $\sigma$ be the request sequence produced. Its length is
\[
|\sigma| = k't < k t = k c \ways s^2 = O(c\ways k \exp(8 (1-\delta)^{-1} \delta^2 \ways)^2),
\]
as required.
\noindent
We will show that $C(\local{\calA}_k, \sigma) \geq c\,C(\calA_{k'}, \sigma) + \omega(1)$ with probability strictly larger than $1/2$, which implies the claim.

Because $\calA_{k'}$ is conservative, it only misses the first time each page in $S_i$ is accessed. Thus, $C(\calA_{k'}, \sigma) = k's$. Now we analyze the cost of $\local{\calA}_k$ on $\sigma$.

\begin{claim}
In each iteration of the loop of step 3, there are at least $2csk$ evictions in $\local{\calA}_k$ with probability at least $1 - \exp(-k/(2\ways s))$.
\end{claim}
\begin{proof}
Fix an execution of step 3-a-i. We model the accesses to items in $S_i$ as a balls-and-bins game, in the usual way: each ball represents an item, and each bin represents a bucket. There are $m = (1-\delta)k$ balls thrown into $n = k / \ways$ bins; these are at least $2$ bins, since $\ways \leq k/2$. Let $h := m/n = (1-\delta)\ways$. Let $\epsilon := 2\delta/(1-\delta)$. Notice that 
$\epsilon \leq n - 2$, because this is equivalent to the assumption $(1-\delta)^{-1} \leq k/(2\ways)$. We have
\[
h + \epsilon h = (1 + \epsilon)h = (1 + \epsilon)(1-\delta)\ways > \ways.
\]
In the last inequality we have used the choice of $\epsilon$.
Thus, if a bin is $\epsilon h$-saturated, its associated bucket incurred at least one eviction, regardless of the state of $\local{\calA}_k$ at the beginning of the step. Let $X$ be the number of $\epsilon h$-saturated bins. The number of evictions in $\local{\calA}_k$ during this step is, hence, at least $X$.

Crucially, for a fixed iteration of the loop of step 3 (the outer loop) the random bin choices in the balls-and-bins game of \emph{every} iteration of the loop of step 3-a (the inner loop) are the same. Thus, the number of evictions during a fixed iteration of the outer loop is at least $tX$.

Notice that
\begin{align*}
f(n, m, \epsilon) &= f\left(\frac{k}{\ways}, (1-\delta)k, \frac{2\delta}{1-\delta}\right)\\
&= \frac{k}{\ways} \exp\left(-2 \left(\frac{2\delta}{1-\delta}\right)^2 (1-\delta)\ways\right)\\
&= \frac{k}{\ways} \exp\left(-\frac{8\delta^2 \ways}{1-\delta}\right) = \frac{16 k}{\ways s}.
\end{align*}
\noindent
By \Cref{thm:balls-and-bins-lower-bound}, $X \geq f(n, m, \epsilon)/8 = 2k/(\ways s)$ with probability at least $1 - \exp(-f(n, m, \epsilon)/32) = 1 - \exp(-k/(2\ways s))$. The proof concludes by noting that $t (2k/(\ways s)) = 2csk$.
\end{proof}

Hence, the cost of each iteration of the loop of step 3 is lower bounded by
\begin{align*}
    2csk &\geq csk' + \Omega(k)  = c \, C(\calA_{k'}, \sigma) + \Omega(k),
\end{align*}
with probability $1 - \exp(-k/(2 \ways s))$. Because the $S_i$'s are disjoint, and the hash function is fully independent, this event takes place independently across all iterations of the loop of step 3. Hence, the probability that some iteration of step 3 has cost at least $c \, C(\calA_{k'}, \sigma) + \Omega(k)$ is at least
\[
1 - \exp(-k/(2\ways s))^{s} = 1 - \exp(-k/(2\ways)) \geq 1 - e^{-1} \approx 0.63.
\]
Therefore, $C(\local{\calA}_k, \sigma) \geq c \, C(\calA_{k'}, \sigma) + \omega(1)$ with probability at least $0.63$.

\section{From Long Sequences To All Sequences}
\label{sec:long}

The proof of \Cref{thm:tradeoff} shows that if we wish $\local{\calA}$ to remain competitive on all request sequences, it must be immune to adversarial sequences that exploit repetition. A natural way to attempt this is by providing the algorithm with the ability to draw new hash functions, or \defn{rehashing}. Then, whenever the algorithm deems the current hash function compromised, it transitions to a new one.\footnote{Interestingly, rehashing has been used in recent cache designs to protect from timing attacks that exploit conflict misses~\cite{Qureshi2018Rehashing,Qureshi2019Rehashing}.}

We will rehash every $\poly(k)$ \emph{cache misses}. Notice that running periodic rehashes after some fixed number of accesses (independent of the number of cache misses) is detrimental, because it gives the adversary infinitely many opportunities to exploit load imbalances in the allocations. In fact, the adversary can now simply fix a single set of $k' = (1-\delta)k$ items, and sequentially access every element in the set, ad infinitum.

A downside of rehashing is that it can potentially hurt the paging cost, because it ``cools down'' the buckets. Upon a rehash, buckets must abruptly switch to caching a different set of items, of which they have limited (or no) information about past accesses. When the paging decisions of $\calA$ are largely dependent on the full history of accesses (e.g., in LFU), bucket instances become impaired after a rehash, and may never catch up.

A paging algorithm that does not suffer from this problem is LRU. This is because LRU depends only on the most recent access of the items, so as soon as an item $x$ is accessed for the first time after a rehash, the new bucket assigned to $x$ has the most up-to–date access information about $x$. In the rest of this section we focus on $\calA = \LRU$.

When a new hash function is drawn, the items currently hashed with the old function must be remapped. Consider the following simple method, that we call \defn{full flushing}: When a rehash is triggered, flush the cache and draw a new hash function. If those items are accessed again, they will be placed in their new buckets.

\begin{tbox}
\textbf{Algorithm:} $\ways$-way set-associative LRU with full flushing and cache size $k$

\begin{itemize}
    \item Augment $\ways$-way set-associative LRU with cache size $k$ as follows:
    \begin{itemize}
        \item Keep a counter of the number of cache misses.
        \item Let $d \geq 2$ be an arbitrary constant. Every $k^d$ misses evict all the items in cache, draw a new hash function, and reset the counter.
    \end{itemize}
\end{itemize}
\end{tbox}

We denote this algorithm $\fflocal{\LRU}$.

\begin{theorem}
\label{thm:all-sequences}
Suppose $\ways = \omega(\log k)$. Then, $\ways$-way set-associative LRU with full flushing is $(1+1/\poly(k))$-competitive with LRU w.h.p. using $(1 + \Theta(\sqrt{\log(k)/\ways}))$-resource augmentation on \emph{all} request sequences.
\end{theorem}
\begin{proof}
For an arbitrary constant $c > d$, let $\delta := \sqrt{36 c \ln(k) / \ways}$. Let $k' := (1 - \delta)k$. We will compare the paging cost of $\fflocal{\LRU}_k$ with that of $\LRU_{k'}$.

Partition $\sigma$ into phases as follows. Each phase is a maximal sequence of time steps in which LRU$_{k'}$ incurs exactly one cache miss. (The only miss takes place in the first access of the phase.) The first phase starts at the beginning of $\sigma$, and all other phases starts right after the previous one. Let $\pi_r$ be the $r$-th phase. We further partition each $\pi_r$ into subphases. Each subphase starts either at the beginning of $\pi_r$ or when a rehash is triggered in $\fflocal{\LRU}_k$, and extends until right before the next rehash or the end of $\pi_r$, whatever comes first. We denote $\pi_r^s$ the $s$-th subphase within $\pi_r$.

Let $i_r$ be the time step when $\pi_r$ begins, and let $S_r := \LRU_{k'}(i_r)$. Note that because LRU$_{k'}$ only misses on the very first access of $\pi_r$, the set of items cached by LRU$_{k'}$ is invariantly $S_r$ during the phase.

Divide the bad evictions of $\fflocal{\LRU}_k$ into two classes: those that are due to the flushes, and the rest. We call the latter ones \defn{regular bad evictions}. Let $F$ be the number of bad evictions due to flushes, and let $R$ be the number of regular bad evictions. By \Cref{lem:unstale-evictions},
\begin{equation}
\label{eq:long-sequences-1}
C(\fflocal{\LRU}_k, \sigma) \leq C(\LRU_{k'}, \sigma) + R + F.
\end{equation}
Since the number of rehashes is at most a $1/k^d$-fraction of the total number of cache misses of $\local{LRU}_k$, and every rehash causes $k$ evictions,
\begin{equation}
\label{eq:long-sequences-2}
F \leq k^{-d + 1} C(\fflocal{\LRU}_k, \sigma).
\end{equation}
Combining \Cref{eq:long-sequences-1} and \Cref{eq:long-sequences-2},
\begin{align*}
C(\fflocal{\LRU}_k, \sigma) &\leq \frac{1}{1 - k^{-d + 1}}(C(\LRU_{k'}, \sigma) + R)\\
&\leq (1 + 2k^{-d + 1})(C(\LRU_{k'}, \sigma) + R) \numberthis \label{eq:long-sequences}.
\end{align*}

\begin{claim}
\label{clm:unstale-evictions-phase}
We have $\prob{\text{There is a regular bad eviction during $\pi_r^s$}} \leq k^{-3c}$.\footnote{This claim is similar to \Cref{clm:unstale-eviction-probability}, but it differs in two crucial ways: it deals with rehashes, and it targets a whole subphase (which can be arbitrarily long) as opposed to a single point in time.}
\end{claim}
\begin{proof}
Suppose a regular bad eviction of some item $x$ in $\local{\LRU}_k$ occurs during $\pi_r^s$, say at time step $j$. Then, $x \in \LRU_{k'}(j)$. By definition of LRU, every other item in $x$'s bucket has been accessed more recently than $x$. This implies that all of those items are in LRU$_{k'}(j)$ as well.\footnote{Notice that this part of the proof would fail for LRU-$K$ (with $K \geq 2$) and LFU, because rehashes clear the access history in the set-associative cache.} Thus, there are $\ways+1$ items in LRU$_{k'}(j)$ that hash to the same bucket. Furthermore, since the set of cached items in LRU$_{k'}$ doesn't change during a phase, we have LRU$_{k'}(j) = S_r$, so, in fact, all the items in $x$'s bucket are in $S_r$. Hence,
\begin{align*}
\Pr[\text{There} & \text{ is a regular bad eviction during $\pi_r^s$}]\\
&\leq \Pr[\text{There are $\ways + 1$ items in $S_r$ that hash to the same bucket}].
\numberthis \label{eq:pr-bad-eviction-subphase}
\end{align*}

We upper bound this probability using a balls-and-bins analysis. Crucially, $S_r$ is a deterministic set. The choice of $\delta$ satisfies the hypothesis of \Cref{lem:balls-and-bins}, so the probability that $S_r$ causes a bucket overflow is at most $\exp(-\delta^2 \ways / 12) = k^{-3c}$, as desired.
\end{proof}

\begin{claim}
\label{clm:subphases-rehashing}
If there is no regular bad eviction during $\pi_r^s$, then it's either the first subphase or the last one in $\pi_r$.
\end{claim}
\begin{proof}
We prove the contrapositive. Suppose $\pi_r^s$ is neither the first nor the last subphase. Then, $\pi_r^s$ begins with a rehash, and ends right before the following rehash. Thus, there are $k^d$ cache misses in $\fflocal{\LRU}_k$ during $\pi_r^s$. Because only items from $S_r$ are requested during $\pi_r$, and $\fflocal{\LRU}_k$ misses more than $k > |S_r|$ times during $\pi_r^s$, there exists some item $x \in S_r$ that was evicted by $\fflocal{\LRU}_k$ at least once during $\pi_r^s$.
\end{proof}

Let $R_r$ be the number of regular bad evictions during $\pi_r$.

\begin{claim}
\label{clm:unstale-evictions-subphase}
We have $\expect{R_r} \leq O(k^{-2c})$.
\end{claim}
\begin{proof}
Since every subphase can have at most $k^d$ misses, it can also have at most $k^d$ evictions. In particular, it can have at most $k^d$ regular bad evictions. Thus,
\begin{align*}
\prob{R_r \geq \ell} \leq \Pr[\text{At least $\lceil\ell / k^d\rceil$ subphases in $\pi_r$ have regular bad evictions}].
\end{align*}
If $\lceil\ell/k^d\rceil$ subphases in $\pi_r$ have regular bad evictions, then, by \Cref{clm:subphases-rehashing}, either $\pi_r^1, \pi_r^2, \dots, \pi_r^{\ceiling{\ell/k^d}}$ all have regular bad evictions, or $\pi_r^2, \pi_r^3, \dots, \pi_r^{\ceiling{\ell/k^d} + 1}$ all do, because all intermediate subphases have regular bad evictions. Thus, by union bound,
\begin{align*}
\prob{R_r \geq \ell} &\leq \Pr\Big[\text{$\pi_r^1, \pi_r^2, \dots, \pi_r^{\ceiling{\ell/k^d}}$ all have regular bad evictions}\Big]\\
&\quad + \Pr\Big[\text{$\pi_r^2, \pi_r^3, \dots, \pi_r^{\ceiling{\ell/k^d} + 1}$ all have regular bad evictions}\Big] \numberthis \label{eq:subphases}\\
&\leq 2(k^{-3c})^{\ceiling{\ell/k^d}}. \tag{because subphases use independent hash functions, and by \Cref{clm:unstale-evictions-phase}}
\end{align*}
\noindent
Finally,
\begin{align*}
\expect{R_r} &= \sum_{\ell \geq 1} \prob{R_r \geq \ell} = \sum_{\ell \geq 1} 2(k^{-3c})^{\ceiling{\ell/k^d}}\\
    &= 2k^d \sum_{a \geq 1} (k^{-3c})^a \leq k^{d} O(k^{-3c})\\
    &= O(k^{-2c}) \tag{as $c > d$}.
\end{align*}
\end{proof}

Observe that the number of phases is exactly $C(\LRU_{k'}, \sigma)$. By \Cref{clm:unstale-evictions-subphase} and linearity of expectation, $\expect{R} \leq C(\LRU_{k'}, \sigma) \, O(k^{-2c})$. By Markov's inequality, $R \geq C(\LRU_{k'}, \sigma) \, O(k^{-c})$ with probability $O(k^{-c})$. Plugging this into \Cref{eq:long-sequences}, we get that
\begin{align*}
C(\fflocal{\LRU}_k, \sigma) &\leq (1 + 2k^{-d + 1})(1 + O(k^{-c})) \, C(\LRU_{k'}, \sigma)\\
&\leq (1 + O(k^{-d + 1})) \, C(\LRU_{k'}, \sigma) \tag{as $c > d$}
\end{align*}
with probability $1 - O(k^{-c})$. This concludes the proof.
\end{proof}

\subsection{Incremental Flushing}

The downside of full flushing is that it brings the execution to a halt and performs up to $k$ uninterrupted writes to slow memory. Alternatively, we can deamortize the flushing, spreading out the evictions over time. We call this method \defn{incremental flushing}, and it works as follows: When a rehash is triggered, we draw a new hash function, and immediately resume the execution. At $k$ arbitrary points in time before the next rehash, pick an arbitrary cached item that is still not remapped (i.e., it was cached when the rehash began, and it is still inhabiting its old bucket), and evict it. If during the course of the execution a non-remapped item is hit, we insert it into the new bucket (potentially causing an eviction, if the new bucket is full), and command the LRU instance of the old bucket to delete it. The rehash concludes when there are no more items mapped with the old function; at that point we discard the old hash function. Importantly, because every rehash finishes before the next one begins, there are at most two hash functions in use at all times. On every access, we search the requested item using both of them.

With incremental flushing, however, two complications arise in the competitive analysis from \Cref{thm:all-sequences}, stemming from the fact that there are two hash functions in use simultaneously. The first one is that because the indexing function is changing throughout a subphase, the hash function in the upper bound from \Cref{eq:pr-bad-eviction-subphase} depends on the time at which the bad eviction occurs. Fortunately, this does not require a fix, because no matter the particular moment in time, the changing hash function will always be fully random, which is all we need to further upper bound the probability via the balls-and-bins process.

The second issue is that the random placements \emph{before} a rehash begins are \emph{not} independent of the placements \emph{after} the rehash begins. This means that the random choices in the subphases are no longer independent, and the step following \Cref{eq:subphases} does not hold. There is a simple fix for this: Because every hash function only lives during at most two consecutive subphases, we can bound the probability that a bad eviction occurs on \emph{every other} subphase. Thus, the bound becomes
\[
\prob{R_r \geq \ell} \leq 2 (k^{-3c})^{\ceiling{\ell/k^d} / 2},
\]
which only changes the constants in the end result. The rest of the proof holds without changes.

We denote $\ways$-way set-associative LRU with incremental flushing as $\iflocal{\LRU}$.

\begin{proposition}
\label{prp:all-sequences}
\Cref{thm:all-sequences} holds for $\iflocal{\LRU}$.
\end{proposition}

\subsection{Competitive Analysis Against the Optimal Offline Algorithm}

It is well known that LRU is $(1 + 1/(r-1))$-competitive against the offline optimal fully associative paging algorithm OPT, using $r$-resource augmentation~\cite{Sleator1985Paging}. Combining this with \Cref{thm:all-sequences} and \Cref{prp:all-sequences}, we can compare set-associative LRU caches with OPT.

\begin{proposition}
Suppose $\ways = \omega(\log k)$. For all $r > 1$, $\fflocal{\LRU}$ and $\iflocal{\LRU}$ are $(1 + 1/(r-1) + o(1))$-competitive with \OPT w.h.p. using $(1+o(1))r$-resource augmentation on all request sequences.
\end{proposition}
\begin{proof}
    Let $k' = k / (1 + \Theta(\sqrt{\log(k)/\ways)})$ and $k'' = k' / r$. We will compare $\fflocal{\LRU}_k$ with $\OPT_{k''}$. The resource augmentation is, thus, $k/k'' = (1 + \Theta(\sqrt{\log(k)/\ways})) r$. We have, with high probability,
    \begin{align*}
        C(\fflocal{\LRU}_k, \sigma) &\leq (1 + 1/\poly(k)) \, C(\LRU_{k'}, \sigma) \tag{by \Cref{thm:all-sequences}}\\
        &\leq (1 + 1/\poly(k))(1 + 1/(r-1)) \, C(\OPT_{k''}, \sigma) \tag{by the competitive analysis of \LRU}\\
        &= (1 + 1/(r-1) + 1/\poly(k)) \, C(\OPT_{k''}, \sigma).
    \end{align*}

    For $\iflocal{\LRU}$ the proof is, of course, analogous.
\end{proof}

\noindent
In particular, setting $r = 2$, we conclude that set-associative LRU caches are $(2 + o(1))$-competitive with \OPT w.h.p. using $(2 + o(1))$-resource augmentation on all request sequences.
\section{The Class of Stable Paging Algorithms}
\label{sec:stable-algorithms}

In this section we set out to understand the class of stable algorithms. We wish to answer the following questions: Which known paging algorithms are stable? Which ones are not? More generally, are there necessary and/or sufficient conditions for an algorithm to be stable?

Throughout this section we assume that items come from a totally ordered space. As usual, all the paging algorithms we consider are online.

\subsection{Stack Algorithms}

Stable algorithms are closely related to stack algorithms~\cite{Mattson1970StorageHierarchies}. A paging algorithm $\calA$ is a \defn{stack algorithm} if $\calA_{k}(\sigma) \subseteq \calA_{k + 1}(\sigma)$ for every $\sigma$ and $k$. The main property of stack algorithms is that they avoid Belady's anomaly,\footnote{A paging algorithm $\calA$ has \defn{Belady's anomaly}~\cite{Belady1969Anomaly} if, for some $\sigma$ and integers $a > b$, $C(\calA_a, \sigma) > C(\calA_b, \sigma)$. That is, unexpectedly, the smaller cache performs better than the larger cache on some request sequence.} because there cannot be a cache miss on a cache of size $k+1$ unless there is also a cache miss on a cache of size $k$. Thus, algorithms that experience Belady's anomaly, like FIFO and clock~\cite{Belady1969Anomaly}, are \emph{not} stack algorithms.

What kind of paging algorithms are stack? Roughly speaking, stack algorithms choose eviction victims based on a ranking that does not depend on $k$. This is formalized as follows. Let $\{\preceq_{\sigma}\}_{\sigma \in \universe^*}$ be a family of total orders over $\universe$. We say that a paging algorithm $\calA$ \defn{conforms to $\{\preceq_{\sigma}\}$} if, for every $\tau \in \universe^*$, $z \in \universe$ and $k$, if $x \in \calA_k(\tau) \setminus \evict(\calA_k, \tau, z)$ and $y \in \evict(\calA_k, \tau, z)$, then $x \preceq_{\tau z} y$. When $\calA$ is lazy, this means that the only evicted item $y$ is exactly the largest according to $\preceq_{\tau z}$, among the items in $\calA_k(\tau)$.

The following theorem exhibits the close relationship between stack algorithms and order families. Although this connection is known~\cite{Mattson1970StorageHierarchies}, we are not aware of any previous work with a full proof, so we include it here.

\begin{theorem}
\label{thm:folklore}
Let $\calA$ be a lazy paging algorithm. Then, $\calA$ is a stack algorithm if and only if it there exists a family $\{\preceq_{\sigma}\}$ such that $\calA$ conforms to it.
\end{theorem}

To prove this theorem we will need the following technical lemma.

\begin{restatable*}{lemma}{technical}
\label{lem:technical}
Let $\calA$ be a lazy paging algorithm. Then, for every $\tau$, $k$ and $z$, if
\[
\text{$\calA_k(\tau) \subseteq \calA_{k+1}(\tau)$ and $\evict(\calA_{k+1}, \tau, z) \cap \calA_k(\tau z) \neq \emptyset$,}
\]
then
\[
\evict(\calA_k, \tau, z) \cap \calA_{k+1}(\tau z) \neq \emptyset.
\]
\end{restatable*}

\noindent
The proof of this lemma is deferred to the end of this section.

\begin{proof}[Proof of \Cref{thm:folklore}]
($\Rightarrow$) Let $\Sigma = \{\sigma_1, \dots, \sigma_{|\sigma|}\}$. We will define the $i$-th smallest element with respect to $\preceq_{\sigma}$. Let $s = |\Sigma|$. If $i > s$, the answer is the $(i - s)$-th smallest item in the set $\universe \setminus \Sigma$ of unaccessed items. Otherwise, $i \leq s$. If $i = 1$, the answer is $\sigma_{|\sigma|}$. Otherwise, $1 < i \leq s$. Notice that $|\calA_j(\sigma)| = j$ for all $j \leq s$, as $\calA$ is a lazy algorithm. Combining this with the inclusion property of stack algorithms, we deduce that $|\calA_i(\sigma) \setminus \calA_{i-1}(\sigma)| = 1$. We define the $i$-th smallest element w.r.t. $\preceq_{\sigma}$ as the only item in $\calA_i(\sigma) \setminus \calA_{i-1}(\sigma)$.

This order family satisfies that $\calA_k(\sigma)$ is exactly the set of the $\min\{k, s\}$ smallest items w.r.t. $\preceq_{\sigma}$. This implies that $\calA$ conforms to $\{\preceq_{\sigma}\}$.

($\Leftarrow$) We want to show that $\calA_k(\sigma) \subseteq \calA_{k+1}(\sigma)$ for every $\sigma$ and $k$. We fix $k$ and use induction on $|\sigma|$. The claim is trivial for $|\sigma| = 0$. Let $\sigma = \tau z$, and assume that the claim holds for $\tau$, that is, 
\begin{equation}
\label{eq:folklore-1}
\calA_k(\tau) \subseteq \calA_{k+1}(\tau).
\end{equation}

If $z$ does not cause an eviction in $\calA_{k+1}$, then \Cref{eq:folklore-1} and the fact that $\calA$ is lazy imply that $\calA_k(\tau z) \subseteq \calA_{k+1}(\tau z)$, so we are done.

Otherwise, let $x \in \evict(\calA_{k+1}, \tau, z)$. Since $\calA$ is lazy, no other item was evicted by $\calA_{k+1}$. Hence, by \Cref{eq:folklore-1}, it suffices to show that $x \notin \calA_k(\tau z)$. By way of contradiction, assume $x \in \calA_k(\tau z)$. Then,
\begin{equation*}
\label{eq:folklore-2}
x \in \evict(\calA_{k+1}, \tau, z) \cap \calA_k(\tau z).
\end{equation*}
\noindent
By \Cref{lem:technical}, there exists a
\begin{equation*}
\label{eq:folklore-3}
y \in \evict(\calA_k, \tau, z) \cap \calA_{k+1}(\tau z).
\end{equation*}
Thus, in response to $z$, $\calA_k$ evicted $y$ and not $x$, and $\calA_{k+1}$ evicted $x$ and not $y$. Because $\calA$ conforms to $\preceq_{\sigma}$, this means that $x \preceq_{\sigma} y$ and $y \preceq_{\sigma} x$, a contradiction. This completes the induction.
\end{proof}

\subsection{A Necessary Condition}

Next, we establish the relationship between stable algorithms and stack algorithms.

\begin{theorem}
Let $\calA$ be a lazy paging algorithm. If $\calA$ is stable, then $\calA$ is stack.
\end{theorem}
\begin{proof}
We want to show that $\calA_k(\sigma) \subseteq \calA_{k+1}(\sigma)$ for every $k$ and $\sigma$. The proof's structure is similar to that of \Cref{thm:folklore}. Once again, we fix $k$ and use induction on $|\sigma|$. The case $|\sigma| = 0$ is trivial. Let $\sigma = \tau z$, and assume that the claim holds for $\tau$, that is, $\calA_{k}(\tau) \subseteq \calA_{k+1}(\tau)$. Because $\calA$ is lazy, it suffices to show that if $z$ causes $\calA_{k+1}$ to evict an item $x$, then $x \notin \calA_k(\tau z)$.

Suppose that $x \in \evict(\calA_{k+1}, \tau, z)$ and, to reach a contradiction, that $x \in \calA_k(\tau z)$. Then, by \Cref{lem:technical}, there exists a $y \in \evict(\calA_k, \tau, z) \cap \calA_{k+1}(\tau z)$.

Intuitively, this contradicts the fact that $\calA$ is stable, as the smaller instance $\calA_k$ should have evicted $x$, which is no longer in the larger instance $\calA_{k+1}$, as opposed to $y$, which is still in $\calA_{k+1}$. Formally, by the definition of stable algorithm for $X = \universe$, $a = k + 1$ and $b = k$, we have that $y \in \evict(\calA_k, \tau, z) \cap \calA_{k+1}(\tau) \neq \emptyset$ implies $\calA_k(\tau z) \subseteq \calA_{k+1}(\tau z)$, which gives a contradiction, because $x \in \calA_k(\tau z) \setminus \calA_{k+1}(\tau z)$.
\end{proof}

This theorem provides a necessary condition for lazy stable algorithms, which we can use to rule out some well-known lazy paging algorithms that are not stack algorithms.

\begin{corollary}
FIFO and clock are \emph{not} stable.
\end{corollary}

One may wonder whether stable and stack algorithms are, in fact, the same, at least when restricted to the class of lazy algorithms. We show that this is not the case.

\begin{proposition}
There exists a lazy stack algorithm that is not stable.
\end{proposition}
\begin{proof}
Let $\calR$ be the paging algorithm that always evicts the item $x$ with the largest number of requests between the last two accesses to $x$; this metric is known as the \defn{reuse distance}~\cite{Jiang2002LIRS}. Let $\Phi(\sigma, x)$ be the number of requests between the last $2$ accesses to item $x$ in $\sigma$. We define $\Phi(\sigma, x) = \infty$ if $x$ has not been accessed at least twice. Consider the total order $x \preceq_{\sigma} y$ if and only if $\Phi(\sigma, x) < \Phi(\sigma, y)$ or ($\Phi(\sigma, x) = \Phi(\sigma, y)$ and $x \leq y$). This defines an order family, and $\calR$ conforms to it, by definition. Thus, by \Cref{thm:folklore}, $\calR$ is a stack algorithm.

We claim that $\calR$ is not stable. Let $\universe = \{A, B, C, Y, Z\}$ be the universe of items. We will compare $\calR_3$ and $\calR_4$. Let $\sigma = AYZZZZABYYBC$. Let $X = \{A, B, C, Y\}$ be the set of items read by $\calR_3$. We will show that the definition of stable algorithm is violated in the last access of $\sigma$.

On the one hand, $\calR_3$ reads the sequence $\sigma[X] = AYABYYBC$, and $Y \preceq_{\sigma[X]} A \preceq_{\sigma[X]} B$. Thus, during the final access to $C$, $\calR_3$ evicts $B$. On the other hand, $\calR_4$ reads the whole sequence $\sigma$, and $Y \preceq_{\sigma} Z \preceq_{\sigma} B \preceq_{\sigma} A$. Hence, when $C$ is accessed, $\calR_4$ evicts $A$. Because $\calR_3$ evicted $B$ (which is still in $\calR_4$) and it still contains $A$ (which is not in $\calR_4$), $\calR$ is not stable.
\end{proof}

\subsection{A Sufficient Condition}

Suppose some paging algorithm conforms to a family $\{\preceq_{\sigma}\}$. Notice that, in the extreme case, the orders can change abruptly from one step to the next, that is, $\preceq_{\sigma}$ may be completely different from $\preceq_{\sigma x}$. This is because the orders are simply a rule the algorithm uses to decide which item to evict, and this rule is free to change at any point in time.

We will restrict our attention to order families that do not present this type of instability. We say that $\{\preceq_{\sigma}\}$ is \defn{monotone} if for every $\sigma$, and $x, y, z \in \sigma$ with $y \neq z$,
\[
\text{if $x \preceq_{\sigma} y$, then $x \preceq_{\sigma z} y$.}
\]
In words, a family is monotone when (1) an access can only change the relative order of the accessed item, and (2) the accessed item does not become larger in the order.

\begin{restatable*}{lemma}{monotone}
\label{lem:monotone}
    Let $\calA$ be a lazy paging algorithm that conforms to a monotone family $\{\preceq_{\sigma}\}$. Fix an arbitrary request sequence $\sigma$. Let $n = |\sigma|$, $\Sigma_n = \{\sigma_1, \dots, \sigma_n\}$, and $s = |\Sigma_n|$. Let $x_1, \dots, x_s$ be the elements of $\Sigma_n$ ordered by $\preceq_{\sigma}$. Then, for every $k \leq s$,
\[
\{x_1, \dots, x_{k-1}\} \subseteq \calA_k(\sigma).
\]
\end{restatable*}

\Cref{lem:monotone} says that the cache contents are essentially dictated by the order family, except perhaps for the largest item in the cache. Its proof can be found at the end of this section. The following corollary captures a crucial implication.

\begin{corollary}
\label{cor:monotone}
Let $\calA$ be a lazy paging algorithm that conforms to a monotone family $\{\preceq_{\sigma}\}$. Then, for every $\sigma$, $k$, and $y \in \sigma$, if $y \notin \calA_k(\sigma)$, then there exists at most one $x \in \calA_k(\sigma)$ such that $y \preceq_{\sigma} x$.
\end{corollary}

Recall that stable algorithms have the property that the paging decisions of two instances serving request sequences $\sigma$ and $\sigma[X]$, respectively, are ``entangled''. We wish to capture this feature via order families. We say that $\{\preceq_{\sigma}\}$ is \defn{self-similar} if for every $\sigma$, $X \subseteq \universe$, and $x, y \in \sigma[X]$,
\[
\text{if $x \preceq_{\sigma[X]} y$, then $x \preceq_{\sigma} y$.}
\]

We are now ready to state the main theorem of this section.

\begin{theorem}
\label{thm:monotone}
Let $\calA$ be a lazy paging algorithm. If $\calA$ conforms to a family $\{\preceq_\sigma\}$ that is monotone and self-similar, then $\calA$ is stable.
\end{theorem}
\begin{proof}
Let $\tau \in \universe^*$, $X \subseteq \universe$, $z \in X$, and $a > b$. 
According to the definition of stable algorithm, we wish to show that $\evict(\calA_b, \tau[X], z) \cap \calA_a(\tau z) \neq \emptyset$ implies $\calA_b(\tau[X] z) \subseteq \calA_a(\tau z)$. Assume the antecedent; let $x \in \evict(\calA_b, \tau[X], z) \cap \calA_a(\tau z)$. Let $y \in \calA_b(\tau[X] z)$, and let us show that $y \in \calA_a(\tau z)$.

If $y = z$, there is nothing to show; assume $y \neq z$.
Because $x \in \evict(\calA_b, \tau[X], z)$, we also have $x \neq z$.

Since $y \in \calA_b(\tau[X]z)$, $y \neq z$ and $\calA$ is lazy, we have $y \in \calA_b(\tau[X])$. Hence, in response to $z$, $x$ was evicted by $\calA_b$ while $y$ remained cached, so $y \preceq_{\tau[X]z} x$. Since $y \neq z$, the monotonicity property yields
\begin{equation}
\label{eq:monotone-0}
    y \preceq_{\tau[X]} x.
\end{equation}
Then, the self-similarity property implies
\begin{equation}
\label{eq:monotone-1}
    y \preceq_{\tau} x.
\end{equation}
Since $x \neq z$, we have, by monotonicity,
\begin{equation}
\label{eq:monotone-2}
    y \preceq_{\tau z} x.
\end{equation}

We claim that $y \preceq_{\tau} z$. Since $x \in \evict(\calA_b, \tau[X], z)$ and $\calA$ is lazy, we have $z \notin \calA_b(\tau[X])$. Notice that it cannot be $z \preceq_{\tau[X]} y$, because \Cref{eq:monotone-0} would imply $z \preceq_{\tau[X]} x$, but $z \notin \calA_b(\tau[X])$ cannot be smaller than two items in $\calA_b(\tau[X])$ w.r.t. $\preceq_{\tau[X]}$, by \Cref{cor:monotone}. Hence, $y \preceq_{\tau[X]} z$. By the self-similarity property,
\begin{equation}
\label{eq:monotone-3}
y \preceq_{\tau} z.
\end{equation}

Recall that our goal is to show that $y \in \calA_a(\tau z)$. For the sake of contradiction, assume $y \notin \calA_a(\tau z)$. Since $y \neq z$ and $\calA$ is lazy, we have
\begin{equation}
\label{eq:monotone-4}
y \notin \calA_a(\tau).
\end{equation}
Analogously, since $x \in \calA_a(\tau z)$, $x \neq z$ and $\calA$ is lazy, we have
\begin{equation}
\label{eq:monotone-5}
x \in \calA_a(\tau).
\end{equation}
To avoid contradicting \Cref{cor:monotone}, \Cref{eq:monotone-1}, \Cref{eq:monotone-3}, \Cref{eq:monotone-4} and \Cref{eq:monotone-5} imply that
\begin{equation}
\label{eq:monotone-6}
z \notin \calA_a(\tau),
\end{equation}
and that
\begin{equation}
\label{eq:monotone-7}
\text{$w \preceq_{\tau} y$ for every $w \in \calA_a(\tau) \setminus \{x\}$.}
\end{equation}
Monotonicity implies that
\begin{equation}
\label{eq:monotone-8}
\text{$w \preceq_{\tau z} y$ for every $w \in \calA_a(\tau) \setminus \{x\}$.}
\end{equation}
Thus, by \Cref{eq:monotone-2} and \Cref{eq:monotone-8}, $x$ is the largest item in $\calA_a(\tau)$ w.r.t. $\preceq_{\tau z}$. By \Cref{eq:monotone-6}, the request $z$ was a miss in $\calA_a$ (and the cache was full, because $\calA$ is lazy and there have already been evictions in the past, namely that of $y$). Therefore, $x \in \evict(\calA_a, \tau, z)$, which is a contradiction.
\end{proof}

\subsection{Examples of Stable Algorithms}

\begin{lemma}
LRU-$K$ is stable.
\end{lemma}
\begin{proof}
We will use \Cref{thm:monotone} on $\calA = \LRU$-$K$. Let $\Phi(\sigma, x)$ be the number of requests since the $K$-th most recent access to item $x$ in $\sigma$. We define $\Phi(\sigma, x) = \infty$ if $x$ has been accessed less than $K$ times. Let $x \preceq_{\sigma} y$ if and only if $\Phi(\sigma, x) < \Phi(\sigma, y)$ or ($\Phi(\sigma, x) = \Phi(\sigma, y)$ and $x \leq y$). This defines the order family that LRU-$K$ conforms to. This family is monotone, because accessing an item does not increase its $K$-th most recent access time, and the relative order of all other items remains the same, as their $K$-th most recent access times increase by $1$ after the new access. To conclude the proof we show that $\{\preceq_{\sigma}\}$ is also self-similar.

Suppose $x \preceq_{\sigma[X]} y$ with $x, y \in \sigma[X]$, and let us show that $x \preceq_{\sigma} y$. If $x = y$, the claim is trivial, so assume $x \neq y$. If $\Phi(\sigma[X], x) = \Phi(\sigma[X], y)$, then they must both be equal to $\infty$, because the $K$-th most recent access times of two different items cannot match. Thus, both $x$ and $y$ have been accessed less than $K$ times in $\sigma[X]$ and, consequently, in $\sigma$, so $\Phi(\sigma, x) = \Phi(\sigma, y) = \infty$. Hence, $x \preceq_{\sigma} y$.

Otherwise, $\Phi(\sigma[X], x) \neq \Phi(\sigma[X], y)$. Without loss of generality, assume $\Phi(\sigma[X], x) < \Phi(\sigma[X], y)$. Starting from $\sigma[X]$, we reinsert each of the accesses removed from $\sigma$, one by one. Crucially, none of these accesses are to $x$ or $y$. Since the $K$-th most recent access to $y$ is older than that of $x$, every time we insert an item access in the sequence either both become one unit time older, or only $y$ does. Hence, at the end of the process, the $K$-th most recent access to $y$ is still older than that of $x$, i.e., $\Phi(\sigma, x) < \Phi(\sigma, y)$. Thus, $x \preceq_{\sigma} y$, completing the proof.
\end{proof}

\begin{lemma}
LFU is stable.
\end{lemma}
\begin{proof}
We only sketch the proof, since it is similar to the preceding one. Let $\Phi(\sigma, x)$ be the number of accesses to $x$ in $\sigma$. Let $x \preceq_{\sigma} y$ if and only if $\Phi(\sigma, x) > \Phi(\sigma, y)$ or ($\Phi(\sigma, x) = \Phi(\sigma, y)$ and $x \leq y$). LFU conforms to this order family, which can be shown to be monotone and self-similar. By \Cref{thm:monotone}, LFU is stable.
\end{proof}

\subsection{Missing Proofs}

In the rest of this section we prove \Cref{lem:technical} and \Cref{lem:monotone}. We restate them here for convenience.

\technical
\begin{proof}
In order to streamline the argument, we will not explicitly invoke the laziness of $\calA$ every time we use it. Suppose that
\begin{equation}
\label{eq:technical-1}
\calA_k(\tau) \subseteq \calA_{k+1}(\tau),
\end{equation}
and
\begin{equation}
\label{eq:technical-2}
\evict(\calA_{k+1}, \tau, z) \cap \calA_k(\tau z) \neq \emptyset,
\end{equation}
and let us show that $\evict(\calA_k, \tau, z) \cap \calA_{k+1}(\tau z) \neq \emptyset$.

Let $x \in \evict(\calA_{k+1}, \tau, z) \cap \calA_k(\tau z)$, as per \Cref{eq:technical-2}. Since $\calA_{k+1}$ evicted $x$ to fetch $z$, it had a full cache when $z$ was accessed. Then, $\calA_k$ was also full at that time. Notice that $z$ was a miss in $\calA_k$; otherwise, if $z \in \calA_k(\tau)$, then, by \Cref{eq:technical-1}, $z \in \calA_{k+1}(\tau)$, contradicting the fact that $z$ caused an eviction in $\calA_{k+1}$. Hence, $\calA_k$ was full when it missed $z$, and it must have evicted an item, say
\begin{equation}    
\label{eq:technical-3}
y \in \evict(\calA_k, \tau, z).
\end{equation}

To conclude the proof, we show that $y \in \calA_{k+1}(\tau z)$. \Cref{eq:technical-3} implies that $y \in \calA_k(\tau)$. Then, by \Cref{eq:technical-1},
\begin{equation}
\label{eq:technical-4}
y \in \calA_{k+1}(\tau).
\end{equation}
\noindent
Notice that $y \neq x$, since $x \in \calA_k(\tau z)$ and $y \notin \calA_k(\tau z)$. In particular, $y$ was not evicted by $\calA_{k+1}$. Combining this with \Cref{eq:technical-4}, we conclude that $y \in \calA_{k+1}(\tau z)$, as desired.
\end{proof}

\monotone
\begin{proof}
    Once again, we will use the fact that $\calA$ is lazy without being explicit about it. The proof is by induction on $n$. The case $n = 0$ is trivial. Let $n \geq 1$. Let $\tau = \sigma_1\dots\sigma_{n-1}$. Let $y_1, \dots, y_{|\Sigma_{n-1}|}$ be the items of $\Sigma_{n-1}$ ordered by $\preceq_{\tau}$.
    
    Fix $k \leq s$. By the inductive hypothesis,
    \begin{equation}
        \label{eq:induction}
        \{y_1, \dots, y_{k-1}\} \subseteq \calA_k(\tau).
    \end{equation}
    Since the order family is monotone,
    \begin{equation}
        \label{eq:contained}
        \{x_1, \dots, x_{k-1}\} \subseteq \{y_1, \dots, y_{k-1}\} \cup \{\sigma_n\}.
    \end{equation}
    (Intuitively, when $\sigma_n$ is accessed, it is taken from its current position in the ordered list of items $y_1, y_2, \dots$ and pulled forward; if $\sigma_n$ was not accessed before, it is inserted somewhere in the list.)

    If $\sigma_n \in \calA_k(\tau)$, then \Cref{eq:induction} implies that
    \[
        \{y_1, \dots, y_{k-1}\} \cup \{\sigma_n\} \subseteq \calA_k(\sigma).
    \]
    Thus, by \Cref{eq:contained}, $\{x_1, \dots, x_{k-1}\} \subseteq \calA_k(\sigma)$.

    Otherwise, $\sigma_n \notin \calA_k(\tau)$. By \Cref{eq:induction}, we have $\calA_k(\tau) = \{y_1, \dots, y_{k-1}, z\}$ for some $z$ that is larger than every other item in $\calA_k(\tau)$ w.r.t. $\preceq_{\tau}$. By monotonicity, $z$ is also the largest in $\calA_k(\tau)$ w.r.t. $\preceq_{\sigma}$. Thus, when $\sigma_n$ is accessed, $z$ is evicted, and
    \[
    \calA_k(\sigma) = \{y_1, \dots, y_{k-1}, \sigma_n\}.
    \]
    By \Cref{eq:contained}, we conclude that $\{x_1, \dots, x_{k-1}\} \subseteq \calA_k(\sigma)$, as desired. This completes the induction.
\end{proof}
\section{Conclusion}

What is the smallest associativity $\ways$ a cache designer may use such that hit rate of set-associative cache is similar to that of a fully associative cache? The answer to this decades-old question is to set $\ways$ slightly above logarithmic in the cache size. This performance guarantee, however, only lasts for a bounded amount of time. If the cache designer further wishes the set-associative cache to sustain comparable performance on arbitrarily long executions, periodically rehashing after a polynomial number of cache misses suffices.

\section*{Acknowledgments}

We thank the anonymous reviewers of SPAA '23 for their feedback that helped us improve the presentation of the paper. We gratefully acknowledge support from NSF grants CCF-2118830, CCF-2106827, CSR-1763680, CNS-1938709, CNS-2118620 and CCF-2106999. This research was also supported in part by NSERC.

\bibliographystyle{plain}
\bibliography{bibliography}

\end{document}